\documentclass{mcom-l}

\usepackage[utf8]{inputenc}
\usepackage[T1]{fontenc}
\usepackage{microtype}
\usepackage[dvips]{graphicx}
\usepackage{xcolor}
\usepackage{times}
\usepackage{graphicx}
\usepackage{algorithm,algpseudocode}
\usepackage{algorithmicx}
\usepackage{amsmath}
\usepackage{amssymb}
\usepackage{mathrsfs}
\usepackage{mathtools}
\usepackage{url}
\usepackage{stmaryrd}
\usepackage{caption}
\usepackage{maplestd2e}
\DefineParaStyle{Maple Bullet Item}
\DefineParaStyle{Maple Heading 1}
\DefineParaStyle{Maple Warning}
\DefineParaStyle{Maple Heading 4}
\DefineParaStyle{Maple Heading 2}
\DefineParaStyle{Maple Heading 3}
\DefineParaStyle{Maple Dash Item}
\DefineParaStyle{Maple Error}
\DefineParaStyle{Maple Title}
\DefineParaStyle{Maple Text Output}
\DefineParaStyle{Maple Normal}
\DefineCharStyle{Maple 2D Output}
\DefineCharStyle{Maple 2D Input}
\DefineCharStyle{Maple Maple Input}
\DefineCharStyle{Maple 2D Math}
\DefineCharStyle{Maple Hyperlink}

\newtheorem{theorem}{Theorem}[section]
\newtheorem{lemma}[theorem]{Lemma}

\theoremstyle{definition}
\newtheorem{definition}[theorem]{Definition}
\newtheorem{example}[theorem]{Example}

\theoremstyle{remark}
\newtheorem{remark}[theorem]{Remark}

\numberwithin{equation}{section}

\begin{document}

\title[A variant of van Hoeij's algorithm to compute hypergeometric terms]{A variant of van Hoeij's algorithm to compute hypergeometric term solutions of holonomic recurrence equations}

\author{Bertrand Teguia Tabuguia}

\address{Department of Mathematics, University of Kassel, Heinrich-Plett-Str.40, 34132 Kassel, Germany}

\email{bteguia@mathematik.uni-kassel.de}

\thanks{This work gather results from chapters $5$ and $6$ of author's Ph.D. thesis, \cite{BTphd}. The author is grateful to his advisor Wolfram Koepf for invaluable guidance.}



\begin{abstract}
		Linear homogeneous recurrence equations with polynomial coefficients are said to be holonomic. Such equations have been introduced in the last century for proving and discovering combinatorial and hypergeometric identities. Given a field $\mathbb{K}$ of characteristic zero, a term $a_n$ is called hypergeometric with respect to $\mathbb{K}$, if the ratio $a_{n+1}/a_n$ is a rational function over $\mathbb{K}$. The solutions space of holonomic recurrence equations gained more interest in the 1990s from the well known Zeilberger's algorithm. In particular, algorithms computing the subspace of hypergeometric term solutions which covers polynomial, rational, and some algebraic solutions of these equations were investigated by Marko Petkov\v{s}ek (1993) and Mark van Hoeij (1999). The algorithm proposed by the latter is characterized by a much better efficiency than that of the other; it computes, in Gamma representations, a basis of the subspace of hypergeometric term solutions of any given holonomic recurrence equation, and is considered as the current state of the art in this area. Mark van Hoeij implemented his algorithm in the Computer Algebra System (CAS) Maple through the command \textit{LREtools[hypergeomsols]}.
		
		We propose a variant of van Hoeij's algorithm that performs the same efficiency and gives outputs in terms of factorials and shifted factorials, without considering certain recommendations of the original version. We have implementations of our algorithm for the CASs Maxima and Maple. Such an implementation is new for Maxima which is therefore used for general-purpose examples. Our Maxima code is currently available as a third-party package for Maxima. A comparison between van Hoeij's implementation and ours is presented for Maple 2020. It appears that both have the same efficiency, and moreover, for some particular cases, our code finds results where \textit{LREtools[hypergeomsols]} fails.
\end{abstract}

\maketitle

\section{Introduction}
Let $\mathbb{K}$ be a field of characteristic zero. $\mathbb{K}$ is mostly a finite extension field of the rationals. A hypergeometric term can always be written in the form
\begin{equation}
	C^n\cdot R(n) \cdot h(n), \label{eq1}
\end{equation}
where $C\in\mathbb{K}$, $R(n)\in\mathbb{K}(n)$, and $h(n)$ is a hypergeometric term expressed in terms of factorials and shifted factorials (Pochhammer symbols\footnote{For a given constant $p$, the Pochhammer symbol $(p)_n$ is $1$ if $n=0$ and $p\cdot (p+1)\cdots (p+n-1)$ if $n$ is a positive integer.}) such that $h(n+1)/h(n)\in\mathbb{K}(n)$ is monic (see \cite{cluzeau2006computing}, \cite[Chapter 6]{BTphd}). Notice that the representation $(\ref{eq1})$ is unique if we choose to write Pochhammer terms as $(p)_n$ with the real part of $p$ in a fixed half-open real interval of unit amplitude. This rewriting creates a multiplicative rational function factor that is taken into account when computing $R(n)$. We will consider the interval $\mathcal{I}=(0,1]$ in our algorithm, and say that Pochhammer parts, corresponding to $h(n)$, is taken modulo the integers ($\mathbb{Z}$) with respect to $\mathcal{I}$.

\begin{example} 
$3^nn!$ and $7^n\frac{n^2+1}{n+2}\frac{(1/3)_n}{(3/4)_n}$ are two hypergeometric terms of the form $(\ref{eq1})$.
\end{example}

A homogeneous linear recurrence equation with rational function coefficients is equivalent to a linear recurrence equation with polynomial coefficients, considering multiplication of the equation by the corresponding rational coefficients common denominator. Therefore the algorithms we are dealing with easily adapt to such cases. However, for this paper we will consider recurrence equations of the form
\begin{equation}
	\sum_{i=0}^{d} P_i(n)\cdot a_{n+i} = 0, d\in\mathbb{N}, \label{eq2}
\end{equation}
for the indeterminate sequence $a_n$, and the coefficients $P_i(n)\in\mathbb{K}[n], i=0,\ldots,d$. Two main algorithms were proposed to find all its hypergeometric term solutions.

Petkov{\v{s}}ek's approach presented in \cite{petkovvsek1992hypergeometric}, focuses on the computation of ratios\footnote{$a_{n+1}/a_n$ for a hypergeometric term $a_n$.} of hypergeometric term solutions of $(\ref{eq2})$ and look for formulas afterward. However, his algorithm has an exponential worst-case complexity on the degree of $P_0$ and $P_d$. Thus this approach could not be considered as conclusive for such computations. The commands \textit{solve\_rec} of the CAS Maxima and originally \textit{RSolve} of Mathematica implement Petkov{\v{s}}ek's algorithm.

van Hoeij's approach first described in \cite{van1999finite} is much more efficient, and moreover, finds a basis of the subspace of hypergeometric term solutions of $(\ref{eq2})$. Note that the final output in Petkov{\v{s}}ek's algorithm is not necessarily a basis. Therefore one could figure out the gain of efficiency by observing that many computational cases in Petkov{\v{s}}ek's approach are reduced to one case in van Hoeij's approach. Computational details of van Hoeij's algorithm were more explained and complemented in \cite{cluzeau2006computing}. Another important point to notice in this algorithm is how unnecessary splitting fields that increase the running time during computations are avoided. 
\begin{definition}(see \cite[Definition 9]{van1999finite}, \cite[Definition 8]{cluzeau2006computing})
	A point $p + \mathbb{Z}, p\in\mathbb{K}$ is called finite singularity of $(\ref{eq1})$ if there exists $\tau\in\mathbb{Z}$ such that $p+\tau$ is a root of $P_d(n-d)\cdot P_0(n)$.
\end{definition}
From this definition where invariance modulo the integers is clearly put forward, one can see the connection between finite singularities and our used rewriting of hypergeometric term Pochhammer parts.

The power of van Hoeij's algorithm comes from the following main concepts:
\begin{enumerate}
	\item local types\footnote{This notion was introduced to study the local properties of difference operators at infinity (see \cite{Barkatou, Duval})} at infinity of hypergeometric term solutions of $(\ref{eq1})$,
	\item local types or valuation growths of hypergeometric term solutions at finite singularities of $(\ref{eq1})$.
\end{enumerate}
The computations of (1) and (2) constitute the key steps of van Hoeij's algorithm and that is where our approach proceeds differently.
\begin{itemize}
	\item For (1), van Hoeij's algorithm uses the Newton polygon algorithm whereas we use a method based on asymptotic expansion inspired by Petkov{\v{s}}ek's algorithm Poly (see \cite{petkovvsek1992hypergeometric}).
	\item Computing (2) is inherent in van Hoeij's algorithm, but in our approach this is automatically considered in the way we construct $h(n)$ in $(\ref{eq1})$ by taking monic factors modulo the integers of $P_0$ and $P_d(n-d)$. This consideration is valid thanks to Petkov{\v{s}}ek's approach.
\end{itemize}
Apart from these essential differences, it is not trivial to notice that both algorithms do the same thing, because their step orderings do not coincide either.

Our first motivation in implementing van Hoeij's algorithm is to use the outputs for power series representations as described in \cite[Section 10.26]{koepf2006computeralgebra}. Indeed, one can represent some holonomic functions as linear combinations of hypergeometric series whose coefficients are hypergeometric term solutions of underlying holonomic recurrence equations. For this purpose, we want outputs ready for evaluation for non-negative integers. This is settled in our approach by considering Pochhammer parts modulo the integers with respect to $\mathcal{I}$. Contrary to van Hoeij's  Maple implementation \textit{LREtools[hypergeomsols]}, this choice is independent of the recurrence equation considered. \textit{LREtools[hypergeomsols]} uses Gamma representations for $h(n)$, and similarly as Pochhammer symbols, evaluation is not defined for negative integers. Since such normalization is not taken into account in \textit{LREtools[hypergeomsols]}, for use in power series computations or combinatory (see \cite{WolfBook}), one may have to shift the initialization which might even be different for each hypergeometric term appearing in the given output. In other cases it may lead to inconvenient results with evaluation like $\Gamma(1/2)=\sqrt{\pi}$. As we mentioned earlier, we want $h(n)$ in terms of factorials and Pochhammer symbols modulo the integers with respect to $\mathcal{I}$; we will say that such a formula is "simple". 

Let us give an example to illustrate all these.

\begin{example}
Consider the following holonomic recurrence equation
\begin{multline}
\displaystyle \,81\,{n}^{3} \big( -2+n \big)  \big( 2592\,{n}^{15}+56592\,{n}^{14}+566784\,{n}^{13}+3438888\,{n}^{12}+14040866\\
\mbox{}\,{n}^{11}+40413165\,{n}^{10}+83014167\,{n}^{9}+118689722\,{n}^{8}+105269208\,{n}^{7}+24761376\,{n}^{6}\\
\mbox{}-78424336\,{n}^{5}-131026944\,{n}^{4}-108917280\,{n}^{3}-54383616\,{n}^{2}-15593472\,n-1990656 \big)\\
\mbox{}\big( -1+n \big) ^{3} \big( 1+2\,n \big) ^{5} a_n - \big( -1+n \big)  \big( 6718464\,{n}^{24}+165722112\,{n}^{23}+1895913216\,{n}^{22}\\
\mbox{}+13287379968\,{n}^{21}+63281637504\,{n}^{20}+213327813888\,{n}^{19}+505402785504\,{n}^{18}\\
\mbox{}+757111794432\,{n}^{17}+271146179476\,{n}^{16}-2121306037512\,{n}^{15}-7223796390373\,{n}^{14}\\
\mbox{}-14217526943124\,{n}^{13}-20381899157262\,{n}^{12}-22697247078996\,{n}^{11}\\
\mbox{}-20140632084597\,{n}^{10}-14388789455784\,{n}^{9}-8294073141060\,{n}^{8}-3843447511168\,{n}^{7}\\
\mbox{}-1418994576624\,{n}^{6}-411122122112\,{n}^{5}-91298680512\,{n}^{4}-14978958336\,{n}^{3}\\
\mbox{}-1708259328\,{n}^{2}-120766464\,n-3981312 \big)  \big( n+1 \big) ^{3} a_{n+1} +32\, \big( 2592\,{n}^{15}\\
\mbox{}+17712\,{n}^{14}+46656\,{n}^{13}+41208\,{n}^{12}-78046\,{n}^{11}-305161\,{n}^{10}\\
\mbox{}-498877\,{n}^{9}-523438\,{n}^{8}-374752\,{n}^{7}-212350\,{n}^{6}-77798\,{n}^{5}\\
\mbox{}-23024\,{n}^{4}-4682\,{n}^{3}-641\,{n}^{2}-53\,n-2 \big)  \big( n+2 \big) ^{3} \big( 3\,n+4 \big) ^{4} a_{n+2} =0
\label{eq3}
\end{multline}
Our Maple and Maxima implementation finds the following output with CPU times $0.110$ and $0.297$ second respectively.
\begin{equation}
	\left\{\frac{{{n\operatorname{!}}^{3}}}{{{{{\left( \frac{1}{3}\right) }_n}}^{4}}\, {{\left( n-1\right) }^{3}}\, {{n}^{6}}}\operatorname{,}\frac{n\, {{\left( 2 n\right) \operatorname{!}}^{5}}}{\left( n-2\right) \, \left( n-1\right) \, {{4}^{5 n}}\, {{n\operatorname{!}}^{4}}}\right\}\label{eq4}.
\end{equation}
Observe that evaluation can be made for non-negative integers. As one can observe, further simplification are made in our implementation to transform Pochhammer symbols into factorials. 

Maple 2019 \textit{LREtools[hypergeomsols]} finds
\begin{equation}
\left[\Gamma \left( n-2 \right)  \left( \Gamma \left( n+1/2 \right)  \right) ^{5}{n}^{2},
{\frac { \left( \Gamma \left( n-1 \right)  \right) ^{2}\Gamma \left( n-2 \right)  \left( n-2 \right) }{ \left( \Gamma \left( n+1/3 \right)  \right) ^{4}{n}^{3}}}\right]\label{eq5},
\end{equation}
with CPU time $0.265$ second. In this output the Gamma terms $\Gamma\left( n-2 \right)$ and $\Gamma\left( n-1 \right)$ cannot be evaluated at $0$. Moreover their arguments differs by $1$, which shows that shifts most be considered before initialization. Note, however, that this remark is only for further use of hypergeometric terms, one can easily show that $(\ref{eq4})$ and $(\ref{eq5})$ are equivalent.

The Maxima (version 5.44) command \textit{solve\_rec} finds

\begin{equation}
{a_n}=\frac{{{\operatorname{\Gamma }\left( \frac{1}{3}\right) }^{4}}\, {{\mathit{\% k}}_1}\, {{\left( n-2\right) \operatorname{!}}^{3}}\, {{3}^{-4 n-8}}\, {{81}^{n}}}{{{n}^{3}}\, {{\operatorname{\Gamma }\left( \frac{3 n+1}{3}\right) }^{4}}}+\frac{{{\mathit{\% k}}_2}\, \left( n-3\right) \operatorname{!} {{n}^{2}}\, {{2}^{5 n+15}}\, {{\operatorname{\Gamma }\left( \frac{2 n+1}{2}\right) }^{5}}}{{{\ensuremath{\pi} }^{\frac{5}{2}}}\, {{32}^{n}}}\mbox{}, \label{eq6}
\end{equation}
with CPU time $70.266$ seconds. As one may expect, finding formulas from the obtained ratios of hypergeometric terms with Petkov\v{s}ek's algorithm is more complicated. Nevertheless, the large computation time here is due to the degrees of the leading and trailing polynomial coefficients of $(\ref{eq3})$.
\end{example}

This paper goes as follows. In the next section we derive an algorithm to compute holonomic recurrence equations satisfied by a list of linearly independent hypergeometric terms. This algorithm is useful to generate examples and observe some properties of hypergeometric terms by doing forward and backward computations. This can also be done using the Maple package \textit{gfun} (see \cite{gfun}), but note that this code is limited to two recurrence equations as input and the strategy used is slightly different.

In Section \ref{sec2} we give more details on how we normalize the Pochhammer parts of hypergeometric terms. This will help to make the description of the main algorithm of this paper self-contained.

Section \ref{sec3} describes our variant of van Hoeij's algorithm which efficiently computes a basis of the subspace hypergeometric term solutions of $(\ref{eq1})$, using the representation $(\ref{eq1})$. The paper ends with some comparisons with existing implementations.

\section{From hypergeometric terms to holonomic recurrence equations}

It is well known that linear combinations of holonomic functions are holonomic (see \cite[Section 10.9, Section 10.16]{koepf2006computeralgebra}, \cite{stanley1980differentiably}). Since hypergeometric terms are holonomic, there exist algorithms to compute a holonomic recurrence equation of least order satisfied by a given linear combination\footnote{Note that this generally reduces to a sum of hypergeometric terms. Therefore the main information here is that the given hypergeometric terms are distinct.} of hypergeometric terms. Throughout this section we assume there exists an algorithm for finding the rational function defined by the ratio of a hypergeometric term (see \cite[Algorithm 2.2]{WolfBook}). The algorithm of this section is a generalization of the case of two given linearly independent hypergeometric terms. Thus, we treat this particular case and by simple analogy we give the general approach for a given list of linearly independent hypergeometric terms.

\subsection{Case of two linearly independent hypergeometric terms}

Let $a_n$ and $b_n$ be two linearly independent hypergeometric terms over $\mathbb{K}$ such that
\begin{equation}
	a_{n+1} = r_1(n) a_{n}~~ \text{and} ~~ b_{n+1} = r_2(n) b_{n}, \label{eq7}
\end{equation}
where $r_1$ and $r_2$ are rational functions in $\mathbb{K}(n)$. As we consider two terms, the order of the recurrence equation sought is $2$, so we are looking for a recurrence equation of the form
\begin{equation}
	P_{2}(n) s_{n+2} + P_{1}(n) s_{n+1} + P_0(n) s_n =0, \label{eq8}
\end{equation}
where $P_{0}, P_{1}, P_{2}$ are polynomials over $\mathbb{K}$, satisfied by $a_{n}$ and $b_{n}$. We must assume that $P_0\cdot P_{2}\neq 0$, otherwise the recurrence equation can be reduced to a first order recurrence relation. Thus finding $(\ref{eq8})$ is equivalent to searching for rational functions $R_2$ and $R_1$ such that
\begin{equation}
	R_2(n) s_{n+2} + R_1(n) s_{n+1} + s_n =0. \label{eq9}
\end{equation}
Using $(\ref{eq7})$, we have
\begin{equation}
	a_{n+2} = r_1(n+1) a_{n+1}~~ \text{and} ~~ b_{n+2} = r_2(n+1) b_{n+1}. \label{eq10}
\end{equation}
By substitution, $a_{n}$ and $b_{n}$ satisfy $(\ref{eq9})$ if and only if
\begin{equation}
	\begin{cases}
		r_1(n+1) R_2 + R_1 = -\frac{1}{r_1(n)}\\[4mm]  
		r_2(n+1) R_2 + R_1 = -\frac{1}{r_2(n)}
	\end{cases}, \label{eq11}
\end{equation}
which is a linear system of two equations with two unknowns in $\mathbb{K}(n)$. Furthermore, a solution exists and is unique since the determinant of the system
\begin{equation}
	r_a(n+1) - r_b(n+1) \neq 0 \label{eq12}
\end{equation}
by assumption. As a linear system of two equations, the exact solution is easy to compute, that is
\begin{eqnarray}
	R_1(n) &=& \frac{r_2(n+1)r_2(n)-r_1(n+1)r_1(n)}{r_1(n)r_2(n)(r_1(n+1)-r_2(n+1))},\\
	R_2(n) &=& \frac{r_1(n)-r_2(n)}{r_1(n)r_2(n)(r_1(n+1)-r_2(n+1))}. \label{eq13}
\end{eqnarray}

Finally, the holonomic recurrence equation sought is found by multiplying the equation $(\ref{eq9})$ by the common denominator of $R_1(n)$ and $R_2(n)$.

\subsection{General case}

Now we want to generalize the above approach for finitely many linearly independent hypergeometric terms. Let $ a_{n}^{[i]}, i=1,\ldots,d$ $(d\geqslant1)$ be $d$ given linearly independent hypergeometric terms over $\mathbb{K}$ such that
\begin{equation}
	a_{n+1}^{[i]} = r_i(n) a_{n}^{[i]},~i=1,\ldots,d, \label{eq14}
\end{equation}
for some rational functions $r_i$. The vector $\left(R_1(n),R_2(n),\ldots,R_{d}(n)\right)^T\in\mathbb{K}(n)^{d}$ of rational coefficients of the recurrence equation
\begin{equation}
	R_d(n) s_{n+d} + R_{d-1}(n) s_{n+d-1} + \ldots + R_1(n) s_{n+1} + s_n = 0 \label{eq15}
\end{equation}
satisfied by each hypergeometric term $a_{n}^{[i]}$, is the unique vector solution $v\in\mathbb{K}(n)^d$ of the matrix system
\begin{equation}
	\left[ \prod_{k=1}^{j-1}r_i(n+k)\right]_{i,j=1,\ldots,d} \cdot v = - \left(\frac{1}{r_i(n)}\right)^T_{i=1,\ldots,d}. \label{eq16}
\end{equation} 
In case there are linearly dependent hypergeometric terms, one can still use this process by replacing the arbitrary constants appearing in the solution of $(\ref{eq16})$ by zero. This is how we implemented this method. The steps of the algorithm can be summarized as follows.

\begin{algorithm}[h!]
	\caption{Compute the holonomic recurrence equation of least order for a given list $L$ of hypergeometric terms}\label{sumhyperRE}
	\begin{algorithmic}[2]
		\Require  A list $L:=[h_1,\ldots,h_d]$ of hypergeometric terms in the variable $n$ and a symbol $a$.
		\Ensure A holonomic recurrence equation in $a_n$ of least order satisfied by the elements in $L$.
		\begin{enumerate}
			\item Let $R:=[r_i(n),\ldots,r_d(n)]$ be the ratios of the elements in $L$.
			\item If there are some irrational functions in $R$ then stop and return FALSE. No holonomic recurrence equation can be found.
			\item Let 
			\[ M:=\left[ \prod_{k=1}^{j-1}r_i(n+k)\right]_{i,j=1,\ldots,d}.\]
			\item Let 
			\[b:=\left(\frac{1}{r_i(n)}\right)^T_{i=1,\ldots,d}.\]
			\item Let $V$ be the solution of the matrix system $M\cdot v = b$.
			\item If there are arbitrary constants in $V$ then substitute those by zero.
			\item Let $RE:=a_n + \sum_{i=1}^{d} V[i]\cdot a_{n+i}$, where $V[i]$ denotes the $i^{\text{th}}$ component in $V$.
			\item Multiply $RE$ by the common denominator of the components of $V$ and return the result with equality to 0, after factoring the coefficients.
		\end{enumerate}
	\end{algorithmic}
\end{algorithm}

\begin{example}
	We implemented this algorithm as \textit{sumhyperRE}. Let us consider Example 4.1 in \cite{petkovvsek1992hypergeometric} and make a backward computation to find the recurrence equation for 
	\[\frac{1}{(n+1)(n+2)}, ~\text{ and }\frac{(-1)^n(2n+3)}{(n+1)(n+2)}.\]
	Our Maxima code gets
	
	\noindent
	\begin{minipage}[t]{4.0em}\color{red}\bfseries
		(\% i1)
	\end{minipage}
	\begin{minipage}[t]{\textwidth}\color{blue}
		sumhyperRE([1/((n+1)*(n+2)), (-1)\^{}n*(2*n+3)/((n+1)*(n+2))],a[n]);
	\end{minipage}
	\[\displaystyle \tag{\% o1} 
	-\left( n+4\right) \, {a_{n+2}}-{a_{n+1}}+\left( n+1\right) \, {a_n}=0,\mbox{}
	\]
	which is the expected result. Next we recover the Fibonacci recurrence from the golden number and its conjugate.
	
	\noindent
	\begin{minipage}[t]{4.000000em}\color{red}\bfseries
		(\% i2)
	\end{minipage}
	\begin{minipage}[t]{\textwidth}\color{blue}
		sumhyperRE([(1-sqrt(5))\^{}n/2\^{}n, (1+sqrt(5))\^{}n/2\^{}n], a[n]);
	\end{minipage}
	\[\displaystyle \tag{\% o2} 
	-{a_{n+2}}+{a_{n+1}}+{a_n}=0\mbox{}
	\]
\end{example}

The latter example illustrates an important point of the algorithm. Indeed, when considering extension fields to determine hypergeometric term solutions, the conjugates of algebraic numbers involved are also part of the solution basis. We will give more details about this in Section \ref{sec3}.

\section{"Simple" formulas for hypergeometric terms}\label{sec2}

Let $a_n$ be a hypergeometric term over a field $\mathbb{K}$ of characteristic zero. Then by definition $r(n):=a_{n+1}/a_n\in\mathbb{K}(n)$. $\mathbb{K}$ is taken as the minimal extension field of $\mathbb{Q}$ where the numerator and the denominator of $r(n)$ split. We wish to write the formula of $a_n$ by means of only numbers appearing in the splitting field of $r(n)$ using factorials, and when not trivially possible, Pochhammer symbols. We want, moreover, that evaluations of the formula for non-negative integers do not produce other improper operation (usually at $0$) than the division by zero. Nevertheless, the representation we are tagging is mostly valid for evaluations with for non-negative integers. This is what we call a "simple" formula. One could say that a formula is considered to be "simple" when it presents more familiar objects from mathematical dictionaries in a reduced form. In the sense of computing formulas of hypergeometric terms, this consists of simplifying as much as possible, Pochhammer symbols to rational multiples of factorials with positive integer-linear arguments. In this section, we present preliminary steps to recover the representation $(\ref{eq1})$ and gather some classical rules as an algorithm to simplify its Pochhammer part. Similar computations can be found in \cite{koepf1995s}; what is worth to notice is the consideration we make to get "simple" formulas in Section \ref{sec3}.

Consider a rational function
\begin{equation}
	r(k):=\frac{P(k)}{Q(k)},~P(k),Q(k)\in\mathbb{K}[k],~ Q(k)\neq 0~ \text{ for integers } k\geqslant 0 \label{eq17}
\end{equation}
such that $P$ and $Q$ do not have non-negative integer roots. We also assume that roots of $P$ and $Q$ are all distinct. We will see in the next section how our approach prepares all rational functions used to compute hypergeometric term Pochhammer part "simple" formulas to satisfy these assumptions. For example, a rational function $r(k)$ with non-negative integer zeros and poles would implicitly be replaced by $r(k+m)$, where $m=\max\{j\in\mathbb{N}_{\geqslant 0}~:~Q(j)\cdot P(j)=0\}$\footnote{$\mathbb{N}_{\geqslant 0}=\{0,1,2,\ldots\}$}.

We consider a hypergeometric term defined with the property
\begin{equation}
	a_{k+1} = r(k) a_k, ~\text{ for integer } k\geqslant 0. \label{eq18}
\end{equation}
Computing a "simple" formula of such a term is to find its general expression $a_n$ for a positive integer $n$ provided that the corresponding initial values $a_0$ is given. That is the result of the product
\begin{equation}
	\prod_{k=0}^{n-1} r(k). \label{eq19}
\end{equation}
For that purpose, the first step is to split the polynomials of $r$ as follows
\begin{equation}
	r(k)=C\frac{(k+a_1)(k+a_2)\cdots(k+a_p)}{(k+b_1)(k+b_2)\cdots(k+b_q)}, \label{eq20}
\end{equation} 
where $p$ and $q$ are, respectively, the degrees of $P$ and $Q$; $C$ is a constant representing the ratio of the leading coefficients of $P$ and $Q$; and $-a_i$'s, $0\leqslant i\leqslant p$, and $-b_j$'s, $0\leqslant j\leqslant q$ are the zeros and poles of $r$, respectively. From the Pochhammer symbol definition, using $(\ref{eq20})$ it follows that
\begin{equation}
	\prod_{k=0}^{n-1} r(k) = C^n \frac{(a_1)_n(a_2)_n\cdots(a_p)_n}{(b_1)_n(b_2)_n\cdots(b_q)_n}. \label{eq21}
\end{equation}
Thus we are called to try simplifications of ratios and products of Pochhammer symbols, and some isolated ones. Many such computations can be found in books or undergraduate courses, see for example \cite[Exercises 1.1 - 1.5]{WolfBook}. Bellow we recall some classical ones.

$x$ and $y$ denote some numbers, and $j$ an integer.
\begin{description}
	\item[Isolated Rule] Assume $x$ is rational, then one can simplify $(x)_n$, according to the following cases.
			\begin{enumerate}
				\item if $x\in\mathbb{N}$, then
				\begin{eqnarray}
					(x)_n &=& x\cdot(x +1)\cdots(x+n-1)\nonumber\\
					&=& \frac{(x+n-1)!}{(x-1)!} \label{22}
				\end{eqnarray}
				\item Else if $x$ has a denominator equal to $2$, then let $s\in\mathbb{N}$ such that $x=\frac{s}{2}$. $s$ is necessarily an odd integer since $x\notin\mathbb{N}$. We set $s=2t+1$, $t\in\mathbb{N}_{\geqslant0}$, then it follows that
				\begin{eqnarray}
					(x)_n &=& \frac{s}{2}\cdot\left(\frac{s}{2}+1\right)\cdots\left(\frac{s}{2}+n-1\right)\nonumber\\
					&=& \frac{s\cdot(s+2)\cdots(s+2\cdot(n-1))}{2^n}\nonumber\\
					&=& \frac{(2t+1)\cdot(2(t+1)+1)\cdots(2(t+n-1)+1)}{2^n}\nonumber\\
					&=& \frac{\left(2\left(t+n\right)\right)!}{(2t)!\cdot(2t+2)\cdots(2(t+n-1)+2)\cdot 2^n}\nonumber\\
					&=& \frac{\left(2\left(t+n\right)\right)!}{(2t)!\cdot(t+1)\cdots(t+n)\cdot 4^n}\nonumber\\
					&=& \frac{\left(2(t+n)\right)!}{(2t)!4^n\binom{t+n}{n}n!}. \label{eq23}
				\end{eqnarray}
				\item Otherwise no simplification is done for $(x)_n$.
		\end{enumerate}
	
	\item[Ratio Rule] Assume $x-y=j>0$. Then we have
	\begin{eqnarray}
		\frac{(y)_n}{(x)_n} &=&\frac{(y)_j\cdot (y+j)\cdots (y+n-1)}{(y+j)\cdots (y+j+n-1)}\nonumber\\
		&=& \frac{(y)_j}{(y+n)\cdots (y+n+j-1)}\nonumber\\
		&=& \frac{(y)_j}{(y+n)_j}.  \label{eq25}
	\end{eqnarray}
	Therefore for $x-y=j\in\mathbb{Z}$,
	\begin{equation}
		\frac{(y)_n}{(x)_n} = \begin{cases}\frac{(y)_j}{(y+n)_j}~~\text{ if } j>0 \\[3mm] \frac{(x+n)_{-j}}{(x)_{-j}}~~\text{ if } j<0 \end{cases}.\label{eq26}
	\end{equation}
	This shows that differences between zeros and poles of $r$ in $(\ref{eq20})$ should be checked before applying the \textbf{Isolated Rule} in order to apply $(\ref{eq25})$ which can simplify two Pochhammer symbols at the same time. Fortunately, these nice computations can be done in Maxima by combining \textit{makegamma()}, \textit{makefact()}, \textit{minfactorial()} and \textit{factor()} as below.
	
	\noindent
	\begin{minipage}[t]{8ex}\color{red}\bf
		\begin{verbatim}
			(%i1) 
		\end{verbatim}
	\end{minipage}
	\begin{minipage}[t]{\textwidth}\color{blue}
		\begin{verbatim}
			r:pochhammer(7/3,n)/pochhammer(1/3,n);
		\end{verbatim}
	\end{minipage}
	\definecolor{labelcolor}{RGB}{100,0,0}
	\[\displaystyle
	\parbox{10ex}{$\color{labelcolor}\mathrm{\tt (\%o1) }\quad $}
	\frac{{{\left( \frac{7}{3}\right) }_{n}}}{{{\left( \frac{1}{3}\right) }_{n}}}\]
	
	\noindent
	\begin{minipage}[t]{8ex}\color{red}\bf
		\begin{verbatim}
			(%i2) 
		\end{verbatim}
	\end{minipage}
	\begin{minipage}[t]{\textwidth}\color{blue}
		\begin{verbatim}
			factor(minfactorial(makefact(makegamma(r))));
		\end{verbatim}
	\end{minipage}
	\definecolor{labelcolor}{RGB}{100,0,0}
	\[\displaystyle
	\parbox{10ex}{$\color{labelcolor}\mathrm{\tt (\%o2) }\quad $}
	\frac{\left( 1+3\cdot n\right) \cdot \left( 4+3\cdot n\right) }{4}\mbox{}
	\]
	
	\item[Product Rule] We consider the following two rules to simplify $(x)_n(y)_n$.
	\begin{itemize}
		\item Assume $y-x=1/2$, then multiplying $(x)_n$ and $(y)_n=(x+1/2)_n$ by $2^n$ leads to the relation
		\begin{equation}
			(x)_n\cdot \left(x+\frac{1}{2}\right)_n = \frac{(2x)_{2n}}{4^n}. \label{eq27}
		\end{equation}
	    \item Assume $y-x=j>0$, it is easy to see that
	    \begin{equation}
	    	(x)_n \cdot (y)_n = (x)_n\cdot (x+j)_n = \frac{(x)_{n+j}^2}{(x)_j(x+n)_j}.\label{eq28}
	    \end{equation}
	\end{itemize} 
\end{description}

More generally, one can find a hypergeometric term Pochhammer part "simple" formula having ratio with non-negative integer zeros and poles by applying the following algorithm.

\begin{algorithm}[h!]
	\caption{Compute $\prod_{k=1}^{n-1}r(k)$}\label{pochfactorsimp}
	\begin{algorithmic}[2]
		\Require  A rational function $r:=r(n)$ and a variable $n$.
		\Ensure A formula of $\prod_{k=1}^{n-1}r(k)$ in terms of factorial and Pochhammer symbols.
		\begin{enumerate}
			\item Factorize $r$ and write it in terms of linear factors and set 
			\begin{equation*}
				h:=r=C\frac{(n+a_1)(n+a_2)\cdots(n+a_p)}{(n+b_1)(n+b_2)\cdots(n+b_q)}.
			\end{equation*}
			\item Substitute $C$ by $C^n$ in $h$.
			\item For each $a_i$, $i=1,\ldots,p$ do
			\begin{enumerate}
				\item if there is $b_j$ in $h$ such that $a_i-b_j\in\mathbb{Z}$ then substitute $\frac{n+a_i}{n+b_i}$ by applying the \textbf{Ratio Rule} accordingly.
			\end{enumerate}
			\item For the remaining $a_i$'s (resp. $b_j$'s) do
				\begin{enumerate}
					\item if there is $a_{i'}$ (resp. $b_{j'}$) such that $a_i-a_{i'}=\pm 1/2$ or $a_i-a_{i'}\in\mathbb{Z}$ (resp. $a_j-a_{j'}=\pm 1/2$ or $a_j-a_{j'}\in\mathbb{Z}$) then substitute $(n+a_i)(n+a_{i'})$ (resp. $(n+b_j)(n+b_{j'})$) by applying the \textbf{Product Rule} accordingly.
				\end{enumerate}
			\item Substitute the remaining $n+a_i$'s and $n+b_j$'s by the result of the \textbf{Isolated Rule} applied to $(a_i)_n$ and $(b_j)_n$ respectively.
			\item Return $h$.
		\end{enumerate}
	\end{algorithmic}
\end{algorithm}

For a "fair" comparison of efficiency between our implementation and the current Maple \textit{LREtools[hypergeomsols]}, we did not considered the \textbf{Product Rule} in our Maxima and Maple codes (but it will certainly be the case in future versions) since the internal Maple command \textit{LREtools[hypergeomsols]} does not apply simplifications as we do. So these computations are supplementary steps that we use in the algorithm of Section \ref{sec3}. Algorithm \ref{pochfactorsimp} is implemented in our Maxima package as \textit{pochfactorsimp(r,n)}. Below we give some examples.
\clearpage

\begin{example}\item
	
	\noindent
	\begin{minipage}[t]{8ex}\color{red}\bf
		\begin{verbatim}
			(%i1) 
		\end{verbatim}
	\end{minipage}
	\begin{minipage}[t]{\textwidth}\color{blue}
		\begin{verbatim}
			pochfactorsimp(-1/(2*(n+1)*(2*n+1)),n);
		\end{verbatim}
	\end{minipage}
	\definecolor{labelcolor}{RGB}{100,0,0}
	\[\displaystyle
	\parbox{10ex}{$\color{labelcolor}\mathrm{\tt (\%o1) }\quad $}
	\frac{{{\left( -1\right) }^{n}}}{\left( 2\cdot n\right) !}\mbox{}
	\]
	
	\noindent
	\begin{minipage}[t]{8ex}\color{red}\bf
		\begin{verbatim}
			(%i2) 
		\end{verbatim}
	\end{minipage}
	\begin{minipage}[t]{\textwidth}\color{blue}
		\begin{verbatim}
			pochfactorsimp((2*n+3)^2/((n+1)*(2*n+1)),n);
		\end{verbatim}
	\end{minipage}
	\definecolor{labelcolor}{RGB}{100,0,0}
	\[\displaystyle
	\parbox{10ex}{$\color{labelcolor}\mathrm{\tt (\%o2) }\quad $}
	\frac{\left( 1+2\cdot n\right) \cdot {{2}^{n-1}}\cdot \left( 2\cdot \left( 1+n\right) \right) !}{\left( n+1\right) \cdot {{4}^{n}}\cdot {{n!}^{2}}}\mbox{}
	\]
\end{example}

Now that we have described an algorithm for computing "simple" formulas of hypergeometric terms given their ratios with no non-negative integer zeros and poles, let us move to the main algorithm of this paper where such rational functions of the Pochhammer parts are computed.

\section{Basis of hypergeometric term solutions}\label{sec3}

Let us rewrite $(\ref{eq2})$ as follows.
\begin{equation}
	P_d(n)a_{n+d} + P_{d-1}(n)a_{n+d-1}+\cdots P_1(n)a_{n+1}+P_0(n)a_{n}=0, \label{eq29}
\end{equation}
with polynomials $P_i(n)\in\mathbb{K}[n],i=0,\ldots, d$ such that $P_0(n)\cdot P_d(n)\neq 0$. $\mathbb{K}$ is a field of characteristic zero.

We have seen how to compute a holonomic recurrence equation of lowest order satisfied by a given number of linearly independent hypergeometric terms. Any computed hypergeometric term solution of such a holonomic recurrence equation is a linear combination of these linearly independent hypergeometric terms considered. The algorithm of this section is a kind of reverse process which for a given holonomic recurrence equation $(\ref{eq29})$ computes a basis of at most $d$ hypergeometric terms of the space of all hypergeometric term solutions of $(\ref{eq29})$.

In the first place, we establish $(\ref{eq1})$ to see hypergeometric terms in normal forms (see \cite[Chapter 3]{geddes1992algorithms}). Let $a_n, n\in\mathbb{N}_{\geqslant0}$, be a hypergeometric sequence such that $r(n) =a_{n+1}/a_n$ $\in\mathbb{K}(n)$. Then we have 
$$\dfrac{a_1}{a_0}=r(0),~ \dfrac{a_2}{a_1}=r(1),~\ldots,~\dfrac{a_n}{a_{n-1}}=r(n-1),~~n\geqslant1,$$
and therefore
\begin{equation}
	\frac{a_n}{a_0} = \prod_{k=0}^{n-1}\frac{a_{k+1}}{a_k}=\prod_{k=0}^{n-1}r(k) \Rightarrow a_n=a_0\prod_{k=0}^{n-1}r(k) \label{eq30}.
\end{equation}
Factorizing $r(n)$ over $\mathbb{K}$ gives
\begin{equation}
	r(n)= C\dfrac{\prod_{i=1}^{I}(n-\alpha_i)}{\prod_{j=1}^{J}(n-\beta_j)}, \label{eq31}
\end{equation}
%{eq21}
where $C$ is a constant. Note that contrary to $(\ref{eq20})$, in $(\ref{eq31})$ $r(n)$ is considered in a more general setting;  $\alpha_i$ and $\beta_j$ are not uniquely determined and may have negative or positive real parts, which is more general than avoiding non-negative integer values. 

Combining $(\ref{eq30})$ and $(\ref{eq31})$ leads to
\begin{equation}
	a_n = a_0\cdot C^n\cdot \dfrac{(-\alpha_1)_n\cdots(-\alpha_I)_n}{(-\beta_1)_n\cdots(-\beta_J)_n}. \label{eq32}
\end{equation}

Now we want to write each Pochhammer symbol modulo $\mathbb{Z}$ in a certain real interval. That is to say that the real parts of the arguments of Pochhammer terms can be chosen belonging to an interval of amplitude $1$. This is an interesting observation made by van Hoeij. In our case, we choose to rewrite the Pochhammer symbols modulo $\mathbb{Z}$ so that $\alpha_i, \beta_j \in \mathcal{I}:=[-1,0)$. Each Pochhammer symbol is then substituted by a polynomial times another Pochhammer term whose argument differs by an integer $u$. Precisely, let $y$ be a real number (for the case of complex numbers, the computations are applied on their real parts), then its corresponding value in $\mathcal{I}$ is $u=y-\lfloor y\rfloor-1$ and we have
\begin{eqnarray}
	(y)_n &=& \dfrac{(u)_n\cdot(u+n)\cdots(y+n-1)}{u\cdot(u+1)\cdots(y-1)}\nonumber\\
	&=& (u)_n\cdot \dfrac{(u+n)_{y-u}}{(u)_{y-u}}\label{xy1}\\
	&=& \left(y-\lfloor y\rfloor - 1\right)_n \cdot \dfrac{\left(n+y-\lfloor y\rfloor - 1\right)_{\lfloor y\rfloor+1}}{\left(y-\lfloor y\rfloor - 1\right)_{\lfloor y\rfloor+1}}. \label{eq33}
\end{eqnarray}
After applying $(\ref{eq33})$ to each Pochhammer symbol in $(\ref{eq32})$, the remaining expression will have Pochhammer terms having arguments with real parts in $(0,1]$. These terms may have more coincidence than the $(-\alpha_i)_n$ and $(-\beta_j)_n$ in $(\ref{eq32})$ since all Pochhammer terms in $(\ref{eq32})$ whose arguments differ by an integer give the same Pochhammer term modulo $\mathbb{Z}$ after substitution. Therefore there exists a rational function $R(n)\in\mathbb{K}(n)$ and some constant numbers $\tilde{\alpha_1},\ldots,\tilde{\alpha_{I}}$, $\tilde{\beta}_,\ldots,\tilde{\beta_J},$ with real parts in $\mathcal{I},$ such that 
\begin{equation}
	a_n = R(n) \cdot C^n \cdot \dfrac{(-\tilde{\alpha_1})_n\cdots(-\tilde{\alpha_I})_n}{(-\tilde{\beta_1})_n\cdots(-\tilde{\beta_J})_n}. \label{eq34}
\end{equation}
The constant $a_0$ is neglected by linearity since we will look for a basis of hypergeometric term solutions of $(\ref{eq29})$.

Considering multiplicities $e_k$ over $\mathbb{Z}\setminus\{0\}$ and replacing $-\tilde{\alpha_i}$ and $-\tilde{\beta_j}$, by $\theta_k$, we get the normal form
\begin{equation}
	a_n = C^n\cdot R(n) \cdot h(n) := C^n \cdot R(n) \cdot \prod_{k=1}^{K} (\theta_k)_n^{e_k}~(\theta_k\in\mathbb{K}, \text{ with real part in } \mathcal{I}), \label{eq35}
\end{equation}
$K\leqslant \leqslant I+J$. This time all the involved data are uniquely determined. The ratio $r(n)$ can be rewritten as
\begin{equation}
	r(n) = \dfrac{a_{n+1}}{a_n}= \dfrac{R(n+1)}{R(n)} \cdot C \cdot h(n+1)/h(n)= \dfrac{R(n+1)}{R(n)} \cdot C \cdot \prod_{k=1}^{K} (n+\theta_k)^{e_k}~\in\mathbb{K}(n)  \label{eq36}.
\end{equation}

Mark van Hoeij uses Gamma representations in $(\ref{eq35})$ and denotes it singularity structure of $a_n$ (see \cite{van1999finite,cluzeau2006computing}). This representation can be seen as the end point of our algorithm when it computes an element of the basis of hypergeometric terms sought. In fact, the goal of computing a basis of all hypergeometric term solutions of $(\ref{eq29})$ is equivalent to finding solutions of $(\ref{eq29})$ with the structure $(\ref{eq35})$. 

\subsection{Monic factors modulo $\mathbb{Z}$}

\begin{lemma}(\cite[Algorithm Hyper]{petkovvsek1992hypergeometric}, \cite[Left and Right solutions]{cluzeau2006computing})
	The ratio of the Pochhammer part $h(n+1)/h(n)$ of hypergeometric term solutions of $(\ref{eq29})$ are built from monic factors of $P_d(n-d)$ for the numerators and $P_0(n-1)$ for the denominators.
\end{lemma}

This lemma allows us to apply factorization modulo $\mathbb{Z}$ on $P_0(n)$ and $P_d(n)$. In fact the ratio of the Pochhammer part $h(n+1)/h(n)$ in $(\ref{eq36})$ is obtained from factorization of $P_0(n)$ and $P_d(n)$ modulo $\mathbb{Z}$. Let us generate a recurrence equation that will be used while describing the steps of our algorithm.

\noindent
\begin{minipage}[t]{4.000000em}\color{red}\bfseries
	(\% i1)
\end{minipage}
\begin{minipage}[t]{\textwidth}\color{blue}
	\begin{verbatim}
	RE:sumhyperRE([binomial(n+3,n),1/n!,(-1)\^{}n/n,
			(-1)\^{}n/pochhammer(1/2,n)\^{}2],a[n])\$
	\end{verbatim}
\end{minipage}
\vspace{0.25cm}

\noindent
We do not display the output to save space. We will refer to this recurrence equation as $(RE)$. The leading term is 

\noindent
\begin{minipage}[t]{4.000000em}\color{red}\bfseries
	(\% i2)
\end{minipage}
\begin{minipage}[t]{\textwidth}\color{blue}
	first(lhs(RE));
\end{minipage}
\begin{multline*}
\tag{\% 2} 
\left( n+2\right) \, \left( n+3\right) \, \left( n+4\right) \, {{\left( 2 n+7\right) }^{2}} \operatorname{(}64 {{n}^{11}}+1536 {{n}^{10}}+16176 {{n}^{9}}+98080 {{n}^{8}}+377372 {{n}^{7}}\\
+955200 {{n}^{6}}+1584741 {{n}^{5}}+1631354 {{n}^{4}}+852544 {{n}^{3}}-25229 {{n}^{2}}-264212 n-94472\operatorname{)}\, {a_{n+4}},\mbox{}	
\end{multline*}

and the trailing term

\noindent
\begin{minipage}[t]{4.000000em}\color{red}\bfseries
	(\% i3)
\end{minipage}
\begin{minipage}[t]{\textwidth}\color{blue}
	last(lhs(RE));
\end{minipage}
\begin{multline*}
\tag{\% o3}  4 n\, \left( n+4\right)  \operatorname{(}64 {{n}^{11}}+2240 {{n}^{10}}+35056 {{n}^{9}}+323344 {{n}^{8}}+1949788 {{n}^{7}}+8053956 {{n}^{6}}\\
+23188049 {{n}^{5}}+46338535 {{n}^{4}}+62583534 {{n}^{3}}+53821965 {{n}^{2}}+26011175 n+5133154\operatorname{)}\, {a_n}.\mbox{}
\end{multline*}

For simplicity of explanation, we will present computations over the rationals. The case of extension fields works similarly, this choice is just to avoid lengthy notations for roots labeling. 

For the leading polynomial coefficient, the monic factors to be considered after factorization in $\mathbb{Q}$ modulo $\mathbb{Z}$ with roots real parts in $\mathcal{I}$ are
$$(n+1)^{e_1}\left(n + \frac{1}{2}\right)^{e_2},~ \text{ for }0\leqslant e_1 \leqslant 3,~0\leqslant e_2 \leqslant 2.$$
For the trailing term we have
$$(n+1)^e~ \text{ for } 0\leqslant e \leqslant2.$$
Therefore ratios of Pochhammer parts of hypergeometric term solutions are among the following
\begin{align}
	& 1, \frac{1}{(n+1)}, \frac{1}{(n+1)^2}, \frac{1}{(n+1)^3}, \frac{1}{\left(n+\frac{1}{2}\right)}, \frac{1}{\left(n+\frac{1}{2}\right)^2}, (n+1) \nonumber\\
	& \frac{(n+1)}{\left(n+\frac{1}{2}\right)}, \frac{(n+1)}{\left(n+\frac{1}{2}\right)^2}, \frac{(n+1)^2}{\left(n+\frac{1}{2}\right)}, \frac{(n+1)^2}{\left(n+\frac{1}{2}\right)^2} \label{eq37}
\end{align}

Observe that none of these ratios has a non-negative integer zero or pole, hence the type of rational function that we treat with Algorithm \ref{pochfactorsimp}. This is made possible by the fact that we consider factorization modulo $\mathbb{Z}$ with roots real parts in $\mathcal{I}$.

Moreover, not all ratios in $(\ref{eq37})$ should be considered because the exponents of each linear factor appearing in the possible ratios of hypergeometric term solutions can be bounded from the given holonomic recurrence equation. For this purpose van Hoeij's algorithm uses the notion of valuation growth or local types of difference operators at finite singularities \cite[Definition 9]{van1999finite}. Such a point is simply a root modulo $\mathbb{Z}$ of the trailing or the leading polynomial coefficient of $(\ref{eq29})$ as we considered.

Since we are already computing ratios of Pochhammer parts of hypergeometric term solutions, we proceed in a slightly different way than what is described in (\cite{van1999finite, cluzeau2006computing}) to compute exponent bounds at finite singularities. We observed that this reduces to take minimum exponents (or valuations) taken by the corresponding factors modulo $\mathbb{Z}$ in the trailing and the leading polynomial coefficients of the initial recurrence equation as lower bounds. Upper bounds are automatically found while computing ratios. Coming back to our example, it follows that ratios with denominator $(n+1/2)$ should be removed. Therefore the remaining ratios are
\begin{align}
	& 1, \frac{1}{(n+1)}, \frac{1}{(n+1)^2}, \frac{1}{(n+1)^3}, \frac{1}{\left(n+\frac{1}{2}\right)^2}, (n+1) \nonumber\\
	&\frac{(n+1)}{\left(n+\frac{1}{2}\right)^2}, \frac{(n+1)^2}{\left(n+\frac{1}{2}\right)^2}. \label{eq38}
\end{align}

Note that these considerations on monic factors of leading and trailing polynomial coefficients can already be seen as an important efficiency gain when comparing computations with those in \cite{petkovvsek1992hypergeometric} which have to consider more cases. We are going to see more tools in the remaining part of the algorithm that will also filter the ratios in $(\ref{eq38})$.

\subsection{Local type at infinity}

Without ambiguity, we will more often use the terminology "local type" instead of "local type at infinity" since we only consider computations at infinity. This is about a characteristic property of hypergeometric term solutions of holonomic recurrence equations.

We study the behavior of a hypergeometric term ($a_n$) ratio $r(n)$ at infinity. Indeed, at $\infty$ we can write
\begin{equation}
	r(n) = c\cdot n^{\nu}\cdot \left(1+\dfrac{b}{n} + O\left(\dfrac{1}{n^2}\right)\right), \label{eq39}
\end{equation}
with the unique triple $\left(\nu,c,b\right)$ called the local type of $a_n$ at $\infty$.

\begin{theorem}[Fuchs Relations] $\label{fuchs}$ Let $R(n)=\frac{N(n)}{U(n)}$ with $N(n),U(n)$ $\in\mathbb{K}[n]$. The following relations between the local type of a hypergeometric term $a_n$ given by $(\ref{eq35})$ hold:
	\begin{itemize}
		\item[i.] $\nu=\sum_{k=1}^{K}e_k$,
		\item[ii.] $b=\sum_{k=1}^{K}\theta_k~e_k ~ + \deg(N(n)) - \deg(U(n)),$
		\item[iii.]$c=C,$
	\end{itemize}
	where $(\nu,c,b)$ denotes the local type of $a_n$ at $\infty.$
\end{theorem}
\begin{proof}
	From $(\ref{eq36})$ we know that
	\begin{equation}
		r(n)=\dfrac{a_{n+1}}{a_n} = C\cdot \left(\dfrac{R(n+1)}{R(n)}\cdot \prod_{k=1}^{K} (n+\theta_k)^{e_k}\right) \label{eq40}.
	\end{equation}
	We would like to compute a truncated asymptotic expansion of $(\ref{eq40})$. This can be seen as the result of the product of asymptotic expansions of the form $(\ref{eq39})$ of $\frac{R(n+1)}{R(n)}$ and $\prod_{k=1}^{K} (n+\theta_k)^{e_k}$ times $C$. Since $R(n)=\frac{N(n)}{U(n)}$, the highest degree of $n$ in its asymptotic expansion is $\delta = \deg(N(n))-\deg(U(n))$. Hence we read as
	\begin{equation}
		R(n) = c_R \cdot n^{\delta} \cdot \left(1+\dfrac{b_R}{n}+O\left(\dfrac{1}{n^2}\right)\right), \label{eq41}
	\end{equation}
	for some constants $c_R, b_R$. Let us now deduce a truncated asymptotic expansion of $R(n+1)$.
	\begin{eqnarray}
		R(n+1) &=& c_R\cdot(n+1)^{\delta}\cdot \left(1+\dfrac{b_R}{n+1}+O\left(\dfrac{1}{n^2}\right)\right)\nonumber\\
		&=& c_R\cdot n^{\delta} \left(1+\dfrac{1}{n}\right)^{\delta}\cdot \left(1+\dfrac{b_R}{n\left(1+\dfrac{1}{n}\right)}+O\left(\dfrac{1}{n^2}\right)\right)\nonumber\\
		&=& c_R\cdot n^{\delta}\left(1+\dfrac{\delta}{n} + \sum_{j=2}^{\delta}\binom{\delta}{j}\left(\dfrac{1}{n}\right)^j\right)\cdot\left(1+\dfrac{b_R}{n}+O\left(\dfrac{1}{n^2}\right)\right)\nonumber\\
		&=& c_R\cdot n^{\delta}\left(1+\dfrac{b_R+\delta}{n}+O\left(\dfrac{1}{n^2}\right)\right). \label{eq42}
	\end{eqnarray}
	
	Thus from $(\ref{eq41})$ and $(\ref{eq42})$ the first order asymptotic expansion of $\frac{R(n+1)}{R(n)}$ yields
	\begin{eqnarray}
		\frac{R(n+1)}{R(n)} &=& \dfrac{1+\dfrac{b_R+\delta}{n}+O\left(\dfrac{1}{n^2}\right)}{1+\dfrac{b_R}{n}+O\left(\dfrac{1}{n^2}\right)}\nonumber\\
		&=& \left(1+\dfrac{b_R+\delta}{n}+O\left(\dfrac{1}{n^2}\right)\right)\cdot \left(1-\dfrac{b_R}{n}+O\left(\dfrac{1}{n^2}\right)\right)\nonumber\\
		&=& 1+\dfrac{\delta}{n}+O\left(\dfrac{1}{n^2}\right) = 1+\dfrac{\deg(N(n))-\deg(U(n))}{n}+O\left(\dfrac{1}{n^2}\right).\label{eq43}
	\end{eqnarray}
	
	On the other hand
	\begin{eqnarray}
		\left(n+\theta_k\right)^{e_k} &=& n^{e_k}\cdot\left(1+\dfrac{\theta_k}{n}\right)^{e_k}\nonumber\\
		&=& n^{e_k}\cdot\left(1 + \dfrac{\theta_k e_k}{n} + \sum_{j=2}^{e_k}\binom{e_k}{j}\left(\dfrac{\theta_k}{n}\right)^j\right)\nonumber\\
		&=& n^{e_k}\cdot \left(1+\dfrac{\theta_k e_k}{n} + O\left(\dfrac{1}{n^2}\right)\right), \label{eq44}
	\end{eqnarray}
	therefore
	\begin{equation}
		\prod_{k=1}^{K} (n+\theta_k)^{e_k} = n^{\sum_{k=1}^{K}e_k}\cdot \left(1+\dfrac{\sum_{k=1}^{K}\theta_k e_k}{n}+O\left(\dfrac{1}{n^2}\right) \right). \label{eq45}
	\end{equation}
	Finally according to $(\ref{eq40})$, the expansion sought is obtained by the product of $(\ref{eq43})$ and $(\ref{eq45})$ times $C$. That is
	\begin{eqnarray}
		r(n) &=& C\cdot n^{\sum_{k=1}^{K}e_k}\cdot \left(1+\dfrac{\sum_{k=1}^{K}\theta_k e_k}{n}+O\left(\dfrac{1}{n^2}\right) \right)\nonumber\\
		&\phantom{=}& \cdot \left(1+\dfrac{\deg(N(n))-\deg(U(n))}{n}+O\left(\dfrac{1}{n^2}\right)\right)\nonumber\\
		&=& C\cdot n^{\sum_{k=1}^{K}e_k} \left( 1 + \dfrac{\sum_{k=1}^{K}\theta_k e_k + \deg(N(n))-\deg(U(n))}{n} + O\left(\dfrac{1}{n^2}\right)\right), \label{eq46}
	\end{eqnarray}
	from which one easily read off the data of the theorem.
\end{proof}

The first two relations in this theorem tell us that for the local type $(\nu,c,b)$ of a hypergeometric term $a_n$, $\nu$ and $b$ can be found directly from a ratio representing its Pochhammer part. Indeed, observe that modulo $\mathbb{Z}$, the second relation of the theorem reads as
\begin{equation}
	b=\sum_{k=1}^{K}\theta_k~e_k. \label{eq47}
\end{equation}

The third relation will be considered later in this subsection. It is straightforward to find $\nu$ and $b$ corresponding to a hypergeometric term local type from the ratio of its Pochhammer part. Thus this constitute the next step after obtaining ratios as in $(\ref{eq38})$. For example $(n+1)/(n+1/2)^2= n^{-1}(1-1/(4n^2)+O(1/n^3))$ and therefore $\nu=-1$ and $b=-1$ (modulo $\mathbb{Z}$). As mentioned earlier, the map $y\mapsto y-\lfloor y \rfloor -1$ is used to find the correspondence of $y$ modulo $\mathbb{Z}$ in $\mathbb{I}$.

Next, we explain how the local types of hypergeometric term solutions of $(\ref{eq29})$ are computed. This step is considered with the highest priority in our algorithm, because if the set of local types of hypergeometric term solutions of a given holonomic recurrence equation is empty, then there is no hypergeometric term solution over the considered field.

For this step, van Hoeij's algorithm uses the Newton polygon of the difference operator (see \cite[Section 3]{van1999finite}). However, we proceed differently. Our idea is to rewrite $(\ref{eq29})$ for ratios of hypergeometric term solutions, substitute $(\ref{eq39})$ inside, and compute the asymptotic expansion of the nonzero side to find equations for the local types by equating the result to $0$. This process is the same Petkov\v{s}ek used to develop its algorithm Poly (see \cite[Algorithm Poly]{petkovvsek1992hypergeometric}).

Let $a_n$ be a hypergeometric term solution of this equation such that $a_{n+1}=r(n) a_n$ for a rational function $r$. $(\ref{eq29})$ can then be written for $r(n)$ as
\begin{equation}
	\sum_{i=0}^{d}P_i\prod_{j=0}^{i-1} r(n+i) =0. \label{eq48}
\end{equation}
We assume 
\begin{equation}
	r(n)=c\cdot n^{\nu}\cdot\left(1+O\left(\dfrac{1}{n}\right)\right) \label{eq49}
\end{equation}
and we substitute this in $(\ref{eq48})$. Similarly as we did in the proof of Theorem $\ref{fuchs}$, we make computations that yield the possible values of $\nu$ and $c$. If such values are found, say $(\nu_{\text{cand}}, c_{\text{cand}})$, then we rewrite $r(n)$ as
\begin{equation}
	c_{\text{cand}}\cdot n^{\nu_{\text{cand}}}\cdot \left(1+\dfrac{b}{n}+O\left(\dfrac{1}{n^2}\right)\right) \label{eq50}
\end{equation}
and we make new computations to find $b$.

Summarized, our procedure to find local types $(\nu,c,b)$ of hypergeometric term solutions of $(\ref{eq29})$ consists in the following items:
\begin{enumerate}
	\item we compute the possible values for $\nu$;
	\item for each value of $\nu$,
	\begin{itemize}
		\item[2-a] we compute possible values for $c$,
		\item[2-b] for each value found for $c$, we use $\nu$ and $c$ to compute the possible values for $b$; 
		\item[2-c] for each value found for $b$, $(\nu,c,b)$ constitutes a local type of a hypergeometric term solution of $(\ref{eq29})$.
	\end{itemize}
\end{enumerate}

Let us now explain how each value is computed.

\begin{itemize}
	
	\item Computing $\nu$:
	
	Substitute $(\ref{eq49})$ in $(\ref{eq48})$ gives the following terms on the left-hand side 
	\begin{equation}
		c^i\cdot n^{i\cdot\nu}\cdot P_i\cdot \left(1 + O\left(\dfrac{1}{n}\right)\right),~~(0\leqslant i\leqslant d) \label{eq51}
	\end{equation}
	which is equivalent to
	\begin{equation}
		l_i\cdot c^i\cdot n^{i\cdot\nu + \deg(P_i)}\cdot \left(1 + O\left(\dfrac{1}{n}\right)\right), ~~(0\leqslant i\leqslant d)\label{eq52}
	\end{equation}
	where $l_i$ denotes the leading coefficient of $P_i$. Since we are dealing with an equality with right-hand side $0$, the terms having the highest power of $n$ in the asymptotic expansion of the equation left-hand must be zero. However this is only possible if a term of the form $(\ref{eq52})$ has the same power of $n$ with some other terms so that they add to $0$. Therefore we deduce that possible candidates for $\nu$ are integer solutions of linear equations coming from equalities of powers of $n$ for two different terms of the form $(\ref{eq52})$. That is for $0\leqslant i\neq j\leqslant d$, we have the equation
	\begin{equation}
		i\cdot\nu + \deg(P_i) = j\cdot\nu + \deg(P_j) \label{53}
	\end{equation}
	and therefore a possible value for $\nu$ is
	\begin{equation}
		\nu_{i,j} = \dfrac{\deg(P_j) - \deg(P_i)}{i-j}, \label{54}
	\end{equation}
	if the computed value is an integer. 
	
	We then compute $\binom{d}{2}$ such values for $(\ref{eq29})$ and keep the integers. Note that two different couples of terms may give the same value for $\nu$, meaning that the corresponding addition to zero involves all the underlying terms, which is the point of the next item. 
	
	\item Computing $c$:
	
	Assume that we have found a value $\nu_{i,j}\in\mathbb{Z}$ corresponding to $k$ terms in the equation $(\ref{eq48})$ with indices $0\leqslant u_1\neq u_2\neq\ldots\neq u_k \leqslant d$. Then from $(\ref{eq52})$ one easily see that a candidate for $c$ is a solution of the polynomial equation
	\begin{equation}
		l_{u_1}\cdot c^{u_1} + l_{u_2}\cdot c^{u_2} + \cdots+ l_{u_k}\cdot c^{u_k} = 0.\label{eq55}
	\end{equation}
	In fact, since the corresponding terms must add to zero in the asymptotic expansion, their leading coefficients must equal zero. Note that $(\ref{eq55})$ is solved over the considered field $\mathbb{K}$.
	
	Thus,for  each value $c_{i,j}\in\mathbb{K}$ which is a zero of $(\ref{eq55})$ for a given $\nu_{i,j}$, $(\nu_{i,j}, c_{i,j})$ is already a possible couple to be completed for the local type of a hypergeometric term solution of $(\ref{eq29})$.
	
	\item Computing $b$:
	
	For a computed couple $(\nu_{i,j}, c_{i,j})$ as explained above, we rewrite $r(n)$ as 
	\begin{equation}
		c_{i,j}\cdot n^{\nu_{i,j}}\cdot \left(1+\dfrac{b}{n}+O\left(\dfrac{1}{n^2}\right)\right) \label{eq56},
	\end{equation}
	with unknown $b$.  
	
	After substituting $(\ref{eq56})$ in $(\ref{eq48})$ and computing again the asymptotic expansion, terms with highest powers of $n$ add to zero, and therefore the left-hand side of the resulting equation must have a leading term with coefficient as a polynomial in the variable $b$. Since that polynomial must be zero, the possible values for $b$ are its roots. This can be done by computing asymptotic expansion and solve the coefficients equal to zero for the unknown $b$. Finally if we find values for $b\in\mathbb{K}$ then we have found for each $b$ a local type $(\nu,c,b)$ of a possible hypergeometric term solution of $(\ref{eq29})$ over $\mathbb{K}$.			
\end{itemize}

Hence we get the following algorithm.

\begin{algorithm}[h!]
	\caption{Compute local types of all hypergeometric term solutions of $(\ref{eq29})$}\label{localtype}
	\begin{algorithmic}[2]
		\Require  Polynomials
		\begin{equation*}
			P_i(n)	\in \mathbb{K}[n], i=0,\ldots,d \mid P_d(n)\cdot P_0(n)\neq 0
		\end{equation*}
		\Ensure The set of all local types of hypergeometric term solutions of the holonomic RE
		\begin{equation*}
			\sum_{i=0}^{d} P_i(n) a_{n+i} = 0.
		\end{equation*}
		\begin{enumerate}
			
			\item Set $L=\{\}$.
			
			\item For all pairs $\{i,j\}\in\{0,1,\ldots,d\}$, compute
			\begin{equation}
				\nu_{i,j} = \dfrac{\deg(P_j) - \deg(P_i)}{i-j}. \label{eq57}
			\end{equation}
			
			\item For each integer $\nu_{i,j}$ computed in $(\ref{eq57})$, compute the set of solutions in $\mathbb{K}$, say $S_{c,i,j}$, of the polynomial equation
			\begin{equation}
				l_{u_1}\cdot c^{u_1} + l_{u_2}\cdot c^{u_2} + \cdots+ l_{u_j}\cdot c^{u_k} = 0 \label{eq58},
			\end{equation}
			where $l_{u_1},l_{u_2},\ldots,l_{u_k}$ are the leading coefficients of the polynomials $P_{u_1}, P_{u_2},\ldots,P_{u_k}$, $0\leqslant u_1\neq u_2\neq\ldots\neq u_k \leqslant d$ satisfying $(\ref{eq57})$ for the same integer $\nu_{i,j}$.
			\begin{itemize}
				\item[(a)] For each element $c_{i,j}$ of $S_{c,i,j}$ set
				\begin{equation}
					r(n)=c_{i,j}\cdot n^{\nu_{i,j}}\cdot \left(1+\dfrac{b}{n}\right) \label{eq59}.
				\end{equation}
			\end{itemize}
%		\end{enumerate}
%		
%		\algstore{pause4}
%	\end{algorithmic}
%\end{algorithm}
%\clearpage 
%
%\begin{algorithm}[H]
%	\ContinuedFloat
%	\caption{Compute the local types of all hypergeometric term solutions of $(\ref{eq29})$}
%	\begin{algorithmic}[3]
%		\algrestore{pause4}	
%		\State 			 	
%		
%		\begin{itemize}
			\item
			\begin{itemize}	 
				\item[(b)] Compute the coefficient $T_{i,j}(b)$ of the first non-zero term of the asymptotic expansion of
				\begin{equation}
					\sum_{i=0}^{d}P_i\prod_{j=0}^{i-1} r(n+i). \label{eq60}
				\end{equation}
				\item[(c)] Solve $T_{i,j}(b)=0$ in $\mathbb{K}$ for the unknown $b$ and define $S_{b,i,j}$ to be the set of solutions.
				\item[(d)] For each element $b_{i,j}\in S_{b,i,j},$ add the triple $(\nu_{i,j},c_{i,j},b)$ to $L$.
			\end{itemize}
			\item Return $L$.
		  \end{enumerate}
		%\end{itemize}
	\end{algorithmic}
\end{algorithm}

\begin{theorem} \label{theolocaltype} Algorithm $\ref{localtype}$ finds all the local types $(\nu,c,b)$ of all hypergeometric term solutions of $(\ref{eq29})$.
\end{theorem}

\begin{remark}\item
	\begin{itemize}
		\item When extension fields are allowed, Algorithm $\ref{localtype}$ is used to bound the degree of such extensions. Indeed, the computation of $c$ and $b$ for local types defines the degree bound of extension fields to carry throughout the remaining part of the algorithm. This is sometimes important to avoid computations with algebraic numbers whose degrees are larger. This is basically how our algorithm avoid splitting fields from the polynomials of degree $11$ in the trailing and leading coefficients of $(RE)$. More theoretical details on dealing with algebraic extensions can be seen in \cite[Section 8]{cluzeau2006computing}. Regarding implementations, this step is better handled with Maple than Maxima: by its command \textit{RootOf}, Maple is able to manage computations with non-explicit values of algebraic numbers; in Maxima however, such a command is not yet available and so solutions over some algebraic extension fields might be missed with our implementation. Nevertheless, this is not to pretend that our Maple implementation could work on every extension fields, because some computations remains not easily handled. We just want to point out implementation limits that are well considered theoretically. This might be the reason why Maple current \textit{LREtools[hypergeomsols]} seems to be an algebraic-numbers based implementation; we will present examples where \textit{LREtools[hypergeomsols]} gets correct results only when polynomial coefficients are converted to symbolic algebraic numbers using \textit{RootOf}.
		
		\item A further notice about extension fields is a property similar to the conjugate root theorem (see \cite[Lemma 3]{cluzeau2006computing}). Indeed, when a local type $(\nu,c,b)$ where $c$ is defined over an extension field leads to a basis of hypergeometric term solutions, say $B_c$, then local types corresponding to conjugates of $c$ lead to bases of same dimensions as $B_c$ that, only differ from $B_c$ by conjugations of $c$. Therefore an important efficiency can be gained by using this property. At the current time, this is not considered in our implementations but will certainly be the case in future releases.
		
		\item We mention that computations of Algorithm $\ref{localtype}$ can sometimes be used to reduce the number of iterations in the implementation. The important point to notice is that when two linearly independent hypergeometric term solutions have the same local type, Algorithm $\ref{localtype}$ computes it at least twice. Therefore collecting local types in a list might be advantageous, so that when a basis of hypergeometric terms corresponding to a particular local type is found, the latter is discarded from the list of local types. This is a useful tool when the number of computed local types (with repeated values) is less than the order of the given holonomic recurrence equation. Otherwise there is no reason to save local types in a list.
	\end{itemize}
\end{remark}

Thus any ratio candidates whose local type is not in the list of local types (deprived of values for $c$) should not be used in further steps. We implemented a Maxima function \textit{localtype(L,n)} which takes the polynomial coefficients of a holonomic recurrence equation in \textit{L} in the variable \textit{n} and returns a list of triples $[\nu,c,b]$. Applying it for $(RE)$ yields

\noindent
\begin{minipage}[t]{4.000000em}\color{red}\bfseries
	(\% i1)
\end{minipage}
\begin{minipage}[t]{\textwidth}\color{blue}
	localtype(expand(REcoeff(RE,a[n])),n);
\end{minipage}
\[\displaystyle \tag{\% o1} 
[[-2\operatorname{,}-1\operatorname{,}-1]\operatorname{,}[-1\operatorname{,}1\operatorname{,}-1]\operatorname{,}[0\operatorname{,}-1\operatorname{,}-1]\operatorname{,}[0\operatorname{,}1\operatorname{,}-1]].\mbox{}
\]
\textit{REcoeff} is our code to collect coefficients. These are expanded using the Maxima command \textit{expand}. From the obtained output it follows that $1/(n+1)^3$ and $(n+1)$ should also be removed from potential ratios of hypergeometric term solution Pochhammer parts in $(\ref{eq38})$. Note, however, that it is from the computed local types that we get the possible values for $C$ in $(\ref{eq35})$ according to the third relation in Theorem \ref{fuchs}. These will be used in the next step together with their corresponding Pochhammer part ratios.

\subsection{Rational part of hypergeometric terms}

The algorithm goes further in filtering the set of Pochhammer part ratios. Indeed, once we have found all those better candidates for ratios of hypergeometric term solution Pochhammer parts, we need to use again the second Fuchs relation from Theorem \ref{fuchs} in order to find $\delta = \deg(N(n)) - \deg(U(n))$, where $N(n)$ and $U(n)$ are the numerator and the denominator of $R$ in $(\ref{eq36})$. In fact, since we have found values for $b$ and its possible ratio candidates, which means that we can compute $\sum_{k=1}^{K}\theta_k\cdot e_k$, we therefore deduce that these candidates are valid if and only if they satisfy
\begin{equation}
	\delta = b - \sum_{k=1}^{K}\theta_k\cdot e_k \in \mathbb{Z}. \label{eq61}
\end{equation}
 In this case the verification of ratios of Pochhammer parts of hypergeometric term solutions for the value of $b$ should consist in checking if the difference $b$ $-$ $\sum_{k=1}^{K}\theta_k\cdot e_k$ is an integer.

However, $(\ref{eq61})$ can be used in the algorithm only if $b$ is not computed modulo $\mathbb{Z}$. Another approach is to use again asymptotic expansion. Since now we have the ratios with their corresponding values of $c$, according to $(\ref{eq36})$ we can write
\begin{equation}
	r(n) = \dfrac{R(n+1)}{R(n)}\cdot c \cdot \frac{h(n+1)}{h(n)}. \label{eq62}
\end{equation}
Moreover
\begin{equation}
	\dfrac{R(n+1)}{R(n)} = 1 + \dfrac{\delta}{n} + O\left(\dfrac{1}{n^2}\right), \label{eq63}
\end{equation}
where $\delta$ is as in $(\ref{eq61})$. Thus the asymptotic expansion of $(\ref{eq48})$ (left-hand side) with $r(n)$ used by combining $(\ref{eq62})$ and $(\ref{eq63})$ must have a leading term as a polynomial coefficient in the variable $\delta$. Therefore values of $\delta$ are integer roots (if there are some) of that polynomial. If there are not such roots, then the rational function $c \cdot h(n+1)/h(n)$ is removed from the potential hypergeometric term solution Pochhammer parts of $(\ref{eq29})$. Our implementation uses this second approach.

Mostly after this step the number of Pochhammer part ratios of hypergeometric term solution of $(\ref{eq29})$ is considerably reduced or equal to the exact number of hypergeometric term solutions. 

Now, it only remains to find the rational function $R$ in $(\ref{eq35})$ whose a holonomic recurrence equation can be easily computed. Let $c\cdot h(n+1)/h(n)$ be one of the remaining ratios times its corresponding $c$ for the local type. Then the recurrence equation
\begin{equation}
	\sum_{i=0}^{d}P_i\cdot R(n+i)\cdot c^{n+i}\cdot \frac{h(n+1+i)}{h(n+i)}=0, \label{eq64}
\end{equation}
is an equation for the unknown rational function $R(n)$ that we can easily modify to a holonomic recurrence equation. Note, however, that there is no need to use a complete algorithm for computing rational solutions of holonomic recurrence equations. Indeed, since we already have the difference between the degrees of the numerators and the denominators of rational solutions of $(\ref{eq64})$, it is enough to use an algorithm that computes a universal denominator\footnote{A universal denominator of rational solutions of a holonomic recurrence equation is a polynomial that is divisible by all the denominators of rational solutions of that holonomic equation \cite{abramov1999rational}.} $U(n)$ of all rational solutions of $(\ref{eq64})$, and use $\delta$ or its maximum value (for the second approach we proposed) (see $(\ref{eq61})$) to compute a degree bound $\delta + \deg(U(n))$ for the degrees of the corresponding numerators. Substituting $N(n)/U(n)$ in $(\ref{eq64})$ where $N(n)$ is an arbitrary polynomial of degree $\delta + \deg(U(n))$ results in a linear system in the coefficients of the arbitrary polynomial $N(n)$. Finally solving that system gives a basis of all the rational functions $R(n)=N(n)/U(n)$ sought.

Let us then say a few words on the computation of a universal denominator of rational solutions of holonomic recurrence equations. Abramov has proposed most key results for that purpose (see \cite{abramov1998rational, abramov1999rational, abramov2011rational}). A crucial step in Abramov's original algorithm is to compute the dispersion set of two polynomials\footnote{The dispersion set of $A(n)$ and $B(n)$ can be defined as the set of all non-negative integer roots of the resultant polynomial of $A(n)$ and $B(n+h)$ in the variable $h$.}. The dispersion set can efficiently be obtained from full factorization as described in \cite{man1994fast} (see also \cite[Algorithm 5.2]{WolfBook}). We use this method in our implementation of Abramov's algorithm to compute universal denominators. As a little story, note that we could not find neither in Maxima, nor in Maple a satisfactory (in terms of efficiency) implementation for computing dispersion sets. Therefore we implemented the algorithm in \cite{man1994fast} and added it as a by product of our Maxima and Maple packages.

We think this last step of computing the rational function $R(n)$ might sometimes make a difference of efficiency between our algorithm and van Hoeij's original version. In his approach, van Hoeij uses a special algorithm from his idea of finite singularities to determine $R(n)$ (see \cite{van1998rational}). Though we mentioned that a complete algorithm for that purpose is not necessary, the algorithm in \cite{van1998rational} is sometimes suitable with the computations in \cite{van1999finite}. However, comparisons in \cite{abramov2011rational} shows that using our approach or van Hoeij's one at this step does guarantee efficiency gain of one over the other.

\subsection{Our algorithm}

We can now present the complete algorithm of this paper.

\begin{algorithm}[h!]
	\caption{Compute hypergeometric term solutions of $(\ref{eq29})$}\label{myhyper}
	\begin{algorithmic}[2]
		\Require  Polynomials
		\begin{equation*}
			P_i(n)	\in \mathbb{K}(n), i=0,\ldots,d, \mid P_d(n)\cdot P_0(n)\neq 0.\label{eq65}
		\end{equation*}
		\Ensure A basis of all hypergeometric term solutions of the holonomic recurrence equation 
		\begin{equation}
			\sum_{i=0}^{d} P_i(n) a_{n+i}=0 \label{eq66}
		\end{equation}    	
		over $\mathbb{K}$.	    		
		\begin{enumerate}
			\item Set $H=\{\}$.   		
			\item Use Algorithm $\ref{localtype}$ to compute the set $\mathcal{L}$ of all local types at infinity of hypergeometric term solutions of $(\ref{eq66})$.
			\item If $\mathcal{L}=\emptyset$, then stop and return $H$.
		\end{enumerate}
		\algstore{pause3}
	\end{algorithmic}
\end{algorithm}
\clearpage 

\begin{algorithm}[h!]
	\ContinuedFloat
	\caption{Compute hypergeometric term solutions of $(\ref{eq29})$}
	\begin{algorithmic}[3]
		\algrestore{pause3}	
		\State     		
		\begin{itemize}	
			\item[(4)] \label{step4} Construct the set of couple (numerator, denominator)
			\begin{multline}
				\mathcal{P}:= \bigg\lbrace \left(p(n),q(n)\right)\in\mathbb{K}[n]^2~:~ p(n) \text{ and } q(n) \text{ 		are monic factors modulo } \mathbb{Z} \\
				\text{ with roots real parts in } [-1,0)\text{ of } P_0(n-1) \text{ and } P_d(n-d) \text{ respectively} 	\bigg\rbrace,
			\end{multline} 
			for ratio candidates of hypergeometric term solution Pochhammer parts.
			\item[(5)] \label{step5} Remove from $\mathcal{P}$ all couple whose $p(n)$ exponents are less than the minimum multiplicity of the corresponding root modulo $\mathbb{Z}$ in the trailing polynomial coefficient $P_0(n)$. Similarly, clear $\mathcal{P}$ by the same consideration for $q(n)$ exponents and the leading polynomial coefficient $P_d(n)$. Finally substitute each remaining couple $\left(p(n),q(n)\right)$ in $\mathcal{P}$ by $\frac{p(n)}{q(n)}$. 
			\item[(6)] Remove elements in $\mathcal{P}$ that have a larger algebraic degree than the bound given by the local types in $\mathcal{L}$.
			\item[(7)] Construct the set $F_1$ of $c \cdot r$, $r\in \mathcal{P}$ such that $c\cdot r$ has its local type at infinity as an element of $\mathcal{L}$.
				\begin{equation}
					F_1 := \bigg\lbrace c\cdot r ~:~ r = n^{\nu_r}\left(1+\dfrac{b_r}{n} + 	O\left(\dfrac{1}{n^2}\right)\right) \in \mathcal{P} \text{ and } (\nu_r,c, b_n)\in \mathcal{L} \bigg\rbrace.
				\end{equation}
			\item[(8)] Set $F_2:=\{\}$. For each element $f(n)$ of $F_1$  		
			\begin{itemize}
				\item[(a)] Compute a recurrence equation, say $E_{f}$ with the coefficients 
				\begin{equation}
					P_i \cdot \prod_{j=0}^{i}f(n+i),~i=0,\ldots,d,
				\end{equation}
				for the rational function $R(n)$ in $(\ref{eq36})$ of the possible hypergeometric term solutions.
				\item[(b)] Substitute the terms $R(n+i)$ by $(1+\frac{\delta}{n+i})$, $i=0,\ldots,d,$ in $E_f$ and compute the coefficient of the leading term of the asymptotic expansion of the left hand side of $E_f$, say $Q_f(\delta)$.
				\item[(c)] Compute the set $S_{\delta_f}$ of integer roots of $Q_f(\delta)$.
				\item[(d)] If $S_{\delta_f} = \emptyset$ then $f(n)$ is discarded.
				\item[(e)] Else set $\delta_f := \max(S_f)$, rewrite $E_f$ in a holonomic form and add $(f(n), \delta_f,E_f)$ in $F_2$.
			\end{itemize}	    	
			\item[(9)] If $F_2=\emptyset$ then stop and return $H$.
			\item[(10)] For each $(f(n), \delta_f,E_f) \in F_2$
			\begin{itemize}
				\item[(a)] Compute the universal denominator $U_f(n)$ of rational solutions of $E_f$ by using the approach in \cite{man1994fast} to find the needed dispersion set.
				\item[(b)] Update $E_f$ as $E_f'$ with $U_f(n)$ to get a holonomic recurrence for numerators of rational solutions of $E_f$. 
				\item[(c)] Set $d_{N_f}:= \deg(U_f(n)) + \delta_f$, and find a basis of all polynomial solutions of degree at most $d_{N_f}$ of $E_f'$.
				\item[(d)] Use Algorithm $\ref{pochfactorsimp}$ to compute $h_f(n) = \prod_{k=0}^{n-1}f(k)$.
				\item[(e)] For each $N_f(n)\in S_{N_f}$ add $\frac{N_f(n)}{U_f(n)}\cdot h_f(n)$ to $H$.
			\end{itemize}	    
			\item[(11)] Return $H$
		\end{itemize}
	\end{algorithmic}
\end{algorithm} 

We implemented Algorithm $\ref{myhyper}$ in Maxima as \textit{HypervanHoeij(RE,a[n],[K])}, with the default value $Q$ (for rationals\footnote{Rationally valued, not symbolically rational: a parameter declared as rational is not considered as such in the implementation.}) for $K$ representing the field where solutions are computed. One must specify $C$ for $K$ to allow computations over extension fields of $\mathbb{Q}$. Using this implementation to solve $(RE)$ yields

\noindent
\begin{minipage}[t]{4.000000em}\color{red}\bfseries
	(\% i1)
\end{minipage}
\begin{minipage}[t]{\textwidth}\color{blue}
	HypervanHoeij(RE,a[n]);
\end{minipage}
\[\displaystyle \text{Evaluation took } 0.2970 \text{ seconds } (0.3020 \text{ elapsed}) \text{ using } 96.149 \text{MB}.\mbox{}\]
\vspace{-0.5cm}

\[\tag{\% o1} 
\left\{\left( n+1\right) \, \left( n+2\right) \, \left( n+3\right) \operatorname{,}\frac{{{\left( -1\right) }^{n}}}{n}\operatorname{,}\frac{1}{n\operatorname{!}}\operatorname{,}\frac{{{\left( -1\right) }^{n}}\, {{4}^{2 n}}\, {{n\operatorname{!}}^{2}}}{{{\left( 2 n\right) \operatorname{!}}^{2}}}\right\}\mbox{},
\]
with timing (allowed with the Maxima command \textit{showtime}) $0.2970$ second and $96.149$ MB memory used. In Maple, we implemented our algorithm as \textit{rectohyperterm} in our package \textit{FPS}\footnote{Denoting formal power series (FPS), an implementation of main algorithms from \cite{BTphd}.}. Le us do the same computation with our Maple implementation and \textit{LREtools[hypergeomsols]}.

\begin{maplegroup}
	\begin{mapleinput}
		\mapleinline{active}{1d}{RE:=FPS[sumhyperRE]([binomial(n+3,n),1/n!,
		(-1)\symbol{94}n/n,(-1)\symbol{94}n/pochhammer(1/2,n)\symbol{94}2],a(n)):}{}
	\end{mapleinput}
\end{maplegroup}
\begin{maplegroup}
	\begin{mapleinput}
		\mapleinline{active}{1d}{Usage(FPS[rectohyperterm](RE,a(n)))
		}{}
	\end{mapleinput}
	\mapleresult
	memory used=15.94MiB, alloc change=4.00MiB, cpu time=141.00ms, real time=135.00ms, gc time=0ns
	\mapleresult
	
	\begin{maplelatex}
		\[\displaystyle  \left\{  \left( n! \right) ^{-1},{\frac { \left( -1 \right) ^{n}}{n}}, \left( n+3 \right)  \left( n+2 \right)  \left( n+1 \right) ,{\frac { \left( -1 \right) ^{n} \left( n! \right) ^{2}\\
					\mbox{}{16}^{n}}{ \left(  \left( 2\,n \right) ! \right) ^{2}}} \right\} \]
	\end{maplelatex}
\end{maplegroup}
\begin{maplegroup}
	\begin{mapleinput}
		\mapleinline{active}{1d}{Usage(LREtools[hypergeomsols](RE,a(n),\{\},output=basis))
		}{}
	\end{mapleinput}
	\mapleresult
	memory used=21.48MiB, alloc change=25.99MiB, cpu time=343.00ms, real time=277.00ms, gc time=125.00ms
	\mapleresult
	
	\begin{maplelatex}
		\[\displaystyle \left[{n}^{3}+6\,{n}^{2}+11\,n+6,{\frac { \left( -1 \right) ^{n}}{n}}, \left( \Gamma \left( n+1 \right)  \right) ^{-1},{\frac { \left( -1 \right) ^{n}}{ \left( \Gamma \left( 1/2+n \right)  \right) ^{2\\
						\mbox{}}}}\right]\]
	\end{maplelatex}
\end{maplegroup}

The \textit{Usage} command from the \textit{CodeTools} package is used to display timings and memory used. Here one can see the advantage of having an implementation that uses extension fields from user specifications. Allowing extension fields for this example unnecessarily increases the timing (will be the same as for \textit{LREtools[hypergeomsols]}) of computations since $(RE)$ is of order $4$ and we already have $4$ hypergeometric term solutions in the output.

\section{Some comparisons}

Our Maple implementation of the given algorithm was tested on many recurrence equations. Regarding efficiency, the difference between our implementation and the internal Maple 2020 \textit{LREtools[hypergeomsols]} is in the order of milliseconds: for solutions over the rationals, our implementation generally gives a better efficiency; and for solutions over extension fields \textit{LREtools[hypergeomsols]} is generally faster, although the timing are very closed. Nevertheless, there are examples where this comparison does not hold but the timings remains closer. Therefore it is more suitable to say that both implementations have same or similar efficiency, and that it may differ depending on internal commands used. For example, one particular internal computations that slows down our Maple implementation in certain cases is the asymptotic expansion of Algorithm \ref{localtype}, for which Maple commands \textit{series} or \textit{asympt} does not always yield an expansion of the required order. For this reason we repeat the computation (of course, by increasing the order) until the desired order is obtained. Furthermore, to facilitate the steps which follow, the coefficients obtained must often be simplified since at times they are equivalent to zero. This issue does not occur with our Maxima implementation although it usually comes third when we compare efficiencies. But again, there are some examples where this is not verified, sometimes our Maxima code is the fastest; indeed, we think for comparisons involving implementations in Maple and Maxima or two different CAS in general, a first look must be taken at kernels and data structures of both systems. Without going into details on this computer science 'buildings' we only mention that Maple has a C-base kernel whereas Maxima has a Common Lisp-base one, and this could make some differences in the speed of both systems. The reader can visit the programming website \url{https://open.kattis.com} to see how often C programs are the fastest. The author particularly recommends to check the 0-1 Sequences problem.

On the other hand, we have been able to find bugs from Maple 2020 internal command. We give here three unexpected examples where our implementation finds results highlighting the issues encountered.

\begin{maplegroup}
	\begin{mapleinput}
		\mapleinline{active}{1d}{RE1:=FPS[sumhyperRE]([(1-sqrt(7))\symbol{94}n/n!,
			(1+sqrt(7))\symbol{94}n/(n+1)!],a(n))
		}{}
	\end{mapleinput}
	\mapleresult
	\begin{maplelatex}
		\begin{multline*}
			{\it RE1}\, := \, \left( 6\, \sqrt{7}-84\,n-210 \right) a \left( n \right) +4\, \left( n+2 \right)  \left( -7\,n+4\, \sqrt{7}-14 \right) a \left( n+1 \right) \\
			\mbox{}- \left( n+2 \right)  \left( n+3 \right)  \left(  \sqrt{7}-14\,n-21 \right) a \left( n+2 \right) =0
		\end{multline*}
	\end{maplelatex}
\end{maplegroup}
\begin{maplegroup}
	\begin{mapleinput}
		\mapleinline{active}{1d}{LREtools[hypergeomsols](RE1,a(n),\{\},output=basis)
		}{}
	\end{mapleinput}
	\mapleresult
	\begin{maplelatex}
		\[\displaystyle 0\]
	\end{maplelatex}
\end{maplegroup}
which means no solution found;
\begin{maplegroup}
	\begin{mapleinput}
		\mapleinline{active}{1d}{RE2:=FPS[sumhyperRE]([3\symbol{94}(n/5),3\symbol{94}(n/2)],a(n))
		}{}
	\end{mapleinput}
	\mapleresult
	\begin{maplelatex}
		\[\displaystyle {\it RE2}\, := \,3\,a \left( n \right) + \left( -{3}^{4/5}- \sqrt{3} \right) a \left( n+1 \right) +{3}^{3/10}a \left( n+2 \right) =0\]
	\end{maplelatex}
\end{maplegroup}
\begin{maplegroup}
	\begin{mapleinput}
		\mapleinline{active}{1d}{LREtools[hypergeomsols](RE2,a(n),\{\},output=basis)
		}{}
	\end{mapleinput}
	\mapleresult
	\underline{Error, (in mod/Primfield/ReduceField/sort\_B) too many levels of recursion}.
\end{maplegroup}
\vspace{0.25cm}

However, the above issues can be avoided by rewriting algebraic numbers in their Maple standard representation using \textit{convert/RootOf}. This is a temporary workaround proposed by Mark van Heoij, the bug will be corrected in future Maple releases. Hence the reason why we mentioned that \textit{LREtools[hypergeomsols]} is an algebraic-number based implementation. Using our implementation, one gets the expected results 

\begin{maplegroup}
	\begin{mapleinput}
		\mapleinline{active}{1d}{FPS[rectohyperterm](RE1,a(n),C)
		}{}
	\end{mapleinput}
	\mapleresult
	\begin{maplelatex}
		\[\displaystyle  \left\{ {\frac { \left( 1- \sqrt{7} \right) ^{n}}{n!}},{\frac { \left( 1+ \sqrt{7} \right) ^{n}}{ \left( n+1 \right) n!}} \right\} \]
	\end{maplelatex}
\end{maplegroup}
\begin{maplegroup}
	\begin{mapleinput}
		\mapleinline{active}{1d}{FPS[rectohyperterm](RE2,a(n),C)
		}{}
	\end{mapleinput}
	\mapleresult
	\begin{maplelatex}
		\[\displaystyle  \left\{  \left( \sqrt [5]{3} \right) ^{n}, \left(  \sqrt{3} \right) ^{n} \right\}. \]
	\end{maplelatex}
\end{maplegroup}

It is important that the hypergeometric property is kept in implementations so that computations extend to symbolic use of expressions that may not represent numbers. The next example type might be of interest for sequence of functions.

\begin{maplegroup}
	\begin{mapleinput}
		\mapleinline{active}{1d}{RE3:=FPS[sumhyperRE]([ln(x)\symbol{94}n,ln(x*y)\symbol{94}n],a(n))
		}{}
	\end{mapleinput}
	\mapleresult
	\begin{maplelatex}
	\[\displaystyle {\it RE3}\, := \,a \left( n \right) \ln  \left( x \right) \ln  \left( xy \right) + \left( -\ln  \left( x \right) -\ln  \left( xy \right)  \right) a \left( n+1 \right) +a \left( n+2 \right) \\
			\mbox{}=0\]
	\end{maplelatex}
\end{maplegroup}
\begin{maplegroup}
	\begin{mapleinput}
		\mapleinline{active}{1d}{LREtools[hypergeomsols](RE3,a(n),\{\},output=basis)
		}{}
	\end{mapleinput}
	\mapleresult
	\underline{Error, (in mod/Normal/Factored) not implemented}
\end{maplegroup}

\begin{maplegroup}
	\begin{mapleinput}
		\mapleinline{active}{1d}{FPS[rectohyperterm](RE3,a(n),C)
		}{}
	\end{mapleinput}
	\mapleresult
	\begin{maplelatex}
		\[\displaystyle  \left\{  \left( \ln  \left( x \right)  \right) ^{n}, \left( \ln  \left( xy \right)  \right) ^{n} \right\} \]
	\end{maplelatex}
\end{maplegroup}

As with the two first examples, we believe the latter bug can be fixed as well. The point here is to show how having our algorithm, a variant of van Hoeij's algorithm, contributes to ensure and present equivalences of theoretical arguments in \cite{petkovvsek1992hypergeometric,van1999finite, cluzeau2006computing}, and improve cutting edge implementations. Although van Hoeij's algorithm is mentioned at \textit{this footnote link} \footnote{\url{https://reference.wolfram.com/language/tutorial/SomeNotesOnInternalImplementation.html}} to a Mathematica webpage, the implementation in \textit{Rsolve} remains quite slow: we computed the solutions of $(RE)$ and got a much complicated result after about $7$ minutes of computations. It is difficult to decide which algorithm is used as the code is hidden from users. Petkov\v{s}ek's algorithm is the most popular implementation encountered in many CASs. Our result and its implementation bring Maxima to the top level in computing hypergeometric term solutions of holonomic recurrence equations. A software demonstration was recently presented at the International Congress of Mathematical Software (ICMS) 2020. The algorithm of this paper can be used as an important black box for the general case of $m$-interlacings of hypergeometric sequences or $m$-fold hypergeometric terms, $m\in\mathbb{N}$ (see \cite{BTphd, BTWK}).

\bibliographystyle{amsplain}
\providecommand{\bysame}{\leavevmode\hbox to3em{\hrulefill}\thinspace}
\providecommand{\MR}{\relax\ifhmode\unskip\space\fi MR }
\providecommand{\MRhref}[2]{%
	\href{http://www.ams.org/mathscinet-getitem?mr=#1}{#2}
}
\providecommand{\href}[2]{#2}

\end{document}